\documentclass[12pt]{amsart}
\usepackage[margin=1in]{geometry}
\usepackage[colorlinks = true,
            linkcolor = blue,
            urlcolor  = blue,
            citecolor = blue,
            anchorcolor = blue]{hyperref}
\usepackage{mathtools, amsmath, amsthm, amssymb, bbm, enumitem, csquotes, cleveref, tikz}

\usepackage{algorithm}
\usepackage{algorithmic}

\newcommand{\R}{\mathbb{R}} 

\newcommand{\A}{\mathcal{A}}
\newcommand{\M}{\mathcal{M}}
\newcommand{\U}{\mathcal{U}}
\newcommand{\I}{\mathcal{I}}
\newcommand{\F}{\mathcal{F}}

\newcommand{\abs}[1]{\left\lvert#1\right\rvert}

\usepackage{thmtools}
\declaretheorem[name=Theorem]{theorem}
\declaretheorem[sibling=theorem,name=Lemma]{lemma}
\declaretheorem[sibling=theorem,name=Proposition]{proposition}

\declaretheorem[sibling=theorem,name=Definition,style=definition]{definition}
\declaretheorem[sibling=theorem,name=Remark,style=definition]{remark}
\declaretheorem[sibling=theorem,name=Example,style=definition]{example}
\begin{document}
\title{ Extension of Simple Algorithms to the Matroid Secretary Problem }
\author{Simon Park}
\date{\today}
\maketitle
%++++++++++++++++++++++++++++++++++++++++++++++++++++++++++++++++++++++++++++++++++++
\begin{abstract}
Whereas there are simple algorithms that are proven to be optimal for the Classical and the Multiple Choice Secretary Problem, the Matroid Secretary Problem is less thoroughly understood. This paper proposes the generalization of some simple algorithms from the Classical and Multiple Choice versions on the Matroid Secretary Problem. Out of two algorithms that make decisions based on samples, like the Dynkin's algorithm, one is proven to be an instance of Greedy Algorithm \cite{BBSW2022}, while the other is not. A generalized version of the Virtual Algorithm \cite{BIKK2007} obtains a constant competitive ratio for the Hat Graph, the adversarial example for Greedy Algorithms, but fails to do so when a slight modificiation is introduced to the graph. We show that there is no algorithm with Strong Forbidden Sets \cite{STV2021} of size 1 on all graphic matroids. 

\begin{align*}
    &\text{I pledge my honor that this paper represents my own work} \\
    &\text{in accordance with University regulations. /s Juhyun 'Simon' Park}
\end{align*}
\end{abstract}
%++++++++++++++++++++++++++++++++++++++++++++++++++++++++++++++++++++++++++++++++++++
\tableofcontents
%++++++++++++++++++++++++++++++++++++++++++++++++++++++++++++++++++++++++++++++++++++
%++++++++++++++++++++++++++++++++++++++++++++++++++++++++++++++++++++++++++++++++++++
%++++++++++++++++++++++++++++++++++++++++++++++++++++++++++++++++++++++++++++++++++++
\section{Introduction}
The Secretary Problem is a classical problem in online algorithms and optimal stopping theory. We start the paper by introducing the problem setting of the original version of the Secretary Problem and its variants.
%++++++++++++++++++++++++++++++++++++++++++++++++++++++++++++++++++++++++++++++++++++
\subsection{Classical Secretary Problem}
\label{section:CSP}
Consider the following online problem. There is an underlying set (i.e., the universe) $\U$, and each element $u \in \U$ has an associated value $v(u) \in \R_{\geq 0}$ that is \textit{unknown} to the algorithm $\A$ in the beginning. The values $v(u)$ will be revealed to $\A$ one by one in some order, and each time the algorithm receives a new value, it needs to irrevocably decide whether or not to terminate with that value as the final output. The goal of the algorithm is to select the optimal element $u^* \in \U$ that maximizes $v(u^*)$. 

Although the origin of the problem is not clearly known, the name of the problem derives from one of its formulation that was widespread: applicants for a secretary position are being interviewed one by one, and the goal of the interviewer is to accept the one with the highest competency value. 

The optimal solution, known as \emph{Dynkin's algorithm}, is folklore: if we set $\abs{\U} = n$, first sample the first $\lfloor \frac{n}{e} \rfloor$ elements without accepting them. Then use the highest value of $v(u)$ among the samples as a threshold for the remaining elements and accept the first one that has greater value than the threshold. When $n$ goes to infinity, the probability that the algorithm chooses the maximum weight element converges to $\frac{1}{e}$, which is proven to be optimal. \cite{Dynkin}
%++++++++++++++++++++++++++++++++++++++++++++++++++++++++++++++++++++++++++++++++++++
\subsection{Multiple Choice Secretary Problem}
\label{section:MCSP}
Kleinberg \cite{Kleinberg} proposed a natural extension of the secretary problem, where the algorithm can select up to $k$ elements instead of just one. In this case, the goal of the algorithm is to approximate the sum of the values of the elements it chooses as closely as possible to the sum of the values of the top $k$ elements $u_1^*, \cdots, u_k^* \in \U$ (i.e., the elements with the highest values). Babaioff et al. \cite{BIKK2007}  proposed two algorithms that select each of $u_i^*$ with probability $\frac{1}{e}$, which is a stronger condition than just approximating the sum of the values. We present these algorithms below.

\subsubsection{Optimistic Algorithm}
Enumerate the elements of the universe $\U$ as $u_1, \cdots, u_n$ in the order in which they are presented to the algorithm. The optimistic algorithm starts by sampling without selecting the first $t = \lfloor \frac{n}{e} \rfloor$ elements $u_1, \cdots, u_t$. Out of these elements, the algorithm stores the $\min(t, k)$ heaviest elements in a reference set $R = \{ r_1, \cdots, r_{\abs{R}} \}$ in decreasing order of their values. That is, $v(r_1) > \cdots > v(r_{\abs{R}})$. After the sampling phase is over, for each $i > t$, element $u_i$ is accepted if and only if $v(u_i) > u(r_{\abs{R}})$. When an element is accepted, the lowest valued reference element $r_{\abs{R}}$ will be removed from $R$. No item is ever added to $R$ after its initial construction.

To sum up, the algorithm initially uses the $k$-th heaviest element of the samples as the threshold. Once an element is accepted, the threshold moves up to the value of the next heaviest element of the samples. In general, it uses the $(k - i)$-th heaviest element of the samples, where $i$ is the number of accepted elements so far.

\subsubsection{Virtual Algorithm}
\label{algo:Virtual}
During the sampling phase, the virtual algorithm performs the same as the optimistic algorithm and constructs the same reference set $R$. But for each $i > t$, it selects $u_i$ if and only if both 
\begin{enumerate}
    \item $v(u_i) > v(r_{\abs{R}})$
    \item $r_{\abs{R}}$ was one of the samples
\end{enumerate}
Regardless of whether or not $u_i$ was selected, it will be added to $R$ and replace $r_{\abs{R}}$ if $v(u_i) > v(r_{\abs{R}})$. The key idea of the virtual algorithm is to always use the value of the $k$-th best element seen so far as the threshold. But unlike the optimistic algorithm, where a new element with a greater value than the threshold is automatically selected, the virtual algorithm additionally checks if that threshold was set from an element from the samples. Although the algorithm may not sound intuitive, the analysis of the algorithm is far simpler than that of the optimistic algorithm. 
%++++++++++++++++++++++++++++++++++++++++++++++++++++++++++++++++++++++++++++++++++++
\subsection{Matroid Secretary Problem}
Babaioff et al. \cite{BIK2007} further generalized the Multiple Choice Secretary Problem to the Matroid Secretary Problem (MSP). In a MSP, the universe $\U$ is replaced by a matroid $\M = (\U, \I)$, and instead of accepting any set of elements of a fixed size $k$, an algorithm is required to accept an independent set $I \in \I$. In fact, the Multiple Choice Secretary Problem is a special case of MSP where the underlying matroid is a $k$-uniform matroid. 

In the last decade, algorithms with constant competitive ratios have been found for certain types of matroids (e.g., transversal \cite{KRTV2013}, graphic \cite{KP2009}, laminal \cite{MTW2016}), but it is still an open question whether there exists such an algorithm for a general matroid. \cite{BIKK2018} provides a comprehensive summary of the state-of-the-art algorithms for different types of matorids.

One notable negative result is presented in \cite{BBSW2022}. Bahrani et al. showed that a certain class of algorithms that are ``greedy-like'' cannot achieve a constant competitive ratio on general matroids. 
%++++++++++++++++++++++++++++++++++++++++++++++++++++++++++++++++++++++++++++++++++++
\subsection{Overview of the Paper}
The paper starts by introducing basic definitions and preliminary results from matroid theory. We then present the frameworks for the Matroid Secretary Problem from previous works, which will be adopted for the purpose of this paper. Using these frameworks, we analyze the properties of three algorithms that are generalized from simple algorithms for the Classical and the Multiple Choice Secretary Problem.
%++++++++++++++++++++++++++++++++++++++++++++++++++++++++++++++++++++++++++++++++++++
%++++++++++++++++++++++++++++++++++++++++++++++++++++++++++++++++++++++++++++++++++++
%++++++++++++++++++++++++++++++++++++++++++++++++++++++++++++++++++++++++++++++++++++
\section{Preliminaries}
In this section, we review some theories that will be useful for the purpose of this paper.
%++++++++++++++++++++++++++++++++++++++++++++++++++++++++++++++++++++++++++++++++++++
\subsection{Matroid Theory}
\begin{definition}
Given a \emph{universe} or \emph{ground set} $\U$, the tuple $\M = (\U, \I)$ where $\I \subset \mathcal{P}(\U)$ is called a \emph{matroid} if it satisfies the following three properties:
\begin{enumerate}
    \item \textbf{Non-trivial:} $\emptyset \in \I$
    \item \textbf{Downward-closed:} If $S \in \I$ and $T \subset S$, then $T \in \I$
    \item \textbf{Augmentation:} If $S, T \in \I$ and $\abs{S} > \abs{T}$, then $\exists \, i \in S \setminus T$ such that $T \cup \{ i \} \in \I$
\end{enumerate}
Any $I \in \I$ is called an \emph{independent set} and any $I \in \mathcal{P}(\U) \setminus \I$ is called a \emph{dependent set} of $\M$.
\end{definition}

\begin{example}
Given a set $\U$ and $\I = \{ S \subset \U \; \vert \; \abs{S} \leq k \}$, the tuple $\M = (\U, \I)$ is a matroid. A matroid of this type is known as a \emph{$k$-uniform matroid}.
\end{example}

\begin{example}
Given an undirected graph $G = (V, E)$ and the set $\I = \{ E' \subset E \; \vert \; E' \text{ is acyclic} \}$, the tuple $\M = (E, \I)$ is a matroid. A matroid of this type is known as a \emph{graphic matroid}.
\end{example}

\begin{definition}
Given a matroid $\M = (\U, \I)$ and a set $S \subset \U$, the \emph{rank} of $S$ is defined as $rank(S) = \underset{I \in \I, I \subset S}{\max} \abs{I}$.
\end{definition}

\begin{definition}
Given a matroid $\M = (\U, \I)$ and a set $S \subset \U$, the \emph{span} of $S$ is defined as $span(S) = \{ u \in \U \; \vert \; rank(S \cup \{ u \}) = rank(S) \}$.
\end{definition}

\begin{definition}
Given a matroid $\M = (\U, \I)$, a set $B \in \I$ is called a \emph{basis} of $\M$ if there exists no independent set $B'$ such that $B' \supsetneq B$, or equivalently if $span(B) = \U$
\end{definition}

\begin{definition}
Given a matroid $\M = (\U, \I)$ and a set $S \subset \U$, the \emph{restriction} of $\M$ on $S$ is defined as $\M\vert_S = (S, \I\vert_S)$ where $\I\vert_S = \{ I \in \I \; \vert \; I \subset S \} $
\end{definition}

\begin{definition}
Given a matroid $\M = (\U, \I)$ and a set $S \in \I$, the \emph{contraction} of $\M$ by $S$ is defined as $\M / S = (\U \setminus S, \I / S)$ where $\I / S = \{ T \subset \U \setminus S \; \vert \; (S \cup T) \in \I \} $
\end{definition}

\begin{definition}
Given a matroid $\M = (\U, \I)$ and a weight function $v: \U \rightarrow \R_{\geq 0}$, the \emph{max-weight basis} of $\M$ is defined as
\begin{equation*}
MWB(\M) = \underset{I \in \I}{\arg \max} \; v(I)
\end{equation*}
where we abuse notation and denote $v(I) = \sum_{u \in I} v(u)$ for $I \subset \U$. Additionally, if $\I$ can be understood unambiguously from the context, we abuse the notation and denote $MWB(\M)$ as $MWB(\U)$. This notation is most commonly used to denote the max-weight basis $ MWB(\M\vert_S)$ of a subset $S \subset \U$ as $MWB(S)$
\end{definition}

\begin{lemma}
\label{lemma:GreedyAlgo}
Given a matroid $\M = (\U, \I)$, the following properties are well-known to be true and will be presented without proof
\begin{enumerate}
    \item $B$ is a basis of $\M$ if and only if $\abs{B} = rank(\U)$
    \item If $S, T \subset \U$, then $rank(S) + rank(T) \geq rank(S \cup T) + rank(S \cap T)$
    \item If $S \subset \U$ and $u \in \U$, then $rank(S) \leq rank(S \cup \{ u \}) \leq rank(S) + 1$
    \item If $S \subset T \subset \U$, then $span(S) \subset span(T)$
    \item The greedy algorithm \ref{algo:Greedy} finds the max-weight basis of $\M$
    \item The greedy algorithm \ref{algo:Greedy} accepts an element $u_i$ if and only if $u_i \not \in span(\{ u_1, \cdots, u_{i - 1} \})$
\end{enumerate}
\end{lemma}

\begin{algorithm}\caption{Greedy Algorithm for Obtaining a Max-Weight Basis}\label{algo:Greedy}

\begin{flushleft}
Enumerate the elements of $\U = \{ u_1, \cdots, u_n \}$ in decreasing order of their weights.
\end{flushleft}
\begin{algorithmic}
    \STATE Initialize $I \gets \emptyset$
    \FOR{$i = 1$ \TO $n$} 
        \IF{$I \cup \{ u_i \} \in \I$}
            \STATE $I \gets I \cup \{ u_i \}$
        \ENDIF
    \ENDFOR
\end{algorithmic}
\end{algorithm}
%++++++++++++++++++++++++++++++++++++++++++++++++++++++++++++++++++++++++++++++++++++
\subsection{Observations about Matroids}
Using some of the statements in Lemma \ref{lemma:GreedyAlgo}, we will prove  properties of a max-weight basis that will be particularly useful for the analysis of algorithms in MSP in the sections to come. Lemma \ref{lemma:SpanTheSame} says that if one element is spanned by another element, then they span the same set of elements. Lemma \ref{lemma:MWBInSubset} says that the `significance'' of an element is downward-closed in some sense. Any element that is ``significant'' to be included in the max-weight basis for a larger set is necessarily included in the max-weight basis for any subset it appears in. In other words, if we know that an element is not in the max-weight basis of a smaller set, we know that it will not be in the max-weight basis of a larger set. Lemma \ref{lemma:MWBKickOutOne} says that adding one additional element to a set can remove up to one element from the max-weight basis of the original set. Proposition \ref{prop:MWBisMWBofMWB} generalizes the idea that the max-weight basis of a max-weight basis is a max-weight basis. 

\begin{lemma}
\label{lemma:SpanTheSame}
For an independent set $I \in \I$ and $u, u' \in \U \setminus span(I)$ in a matroid $\M = (\U, \I)$, if $u \in span(I \cup \{ u' \} )$, then $span(I \cup \{ u \}) = span(I \cup \{ u' \})$
\end{lemma}
\begin{proof}
Assume $v \not \in span(I \cup \{ u \})$. That is, $I \cup \{ u, v \}$ is independent. Since $I \cup \{ u' \}$ is independent, the augmentation property forces $I \cup \{ u', v \}$ to be independent. Hence, $v \not \in span(I \cup \{ u' \})$. Similarly, if $v \not \in span(I \cup \{ u' \})$, then since $I \cup \{ u \}$ is independent, we need $I \cup \{ u, v \}$ to be independent. Hence, $v \not \in span(I \cup \{ u \})$. 
\end{proof}

\begin{lemma}[Modified from \cite{STV2021}]
\label{lemma:MWBInSubset}
For any set $T \subset S \subset \U$ in a matroid $\M = (\U, \I)$ and any weight function $v: \U \rightarrow \R_{\geq 0}$, we have 
\begin{equation*}
    MWB(S) \cap T \subset MWB(T)
\end{equation*}
In particular, if $T$ contains only the top elements of $S$; that is, if $T$ can be written as $\{ u \in S : v(u) \geq c \}$ for some constant $c$, then
\begin{equation*}
    MWB(S) \cap T = MWB(T)
\end{equation*}
\end{lemma}
\begin{proof}
Label the elements of $S = \{ u_1, \cdots, u_n \}$ in decreasing order of their values. Assume $u_k \in MWB(S) \cap T$. Then since it was accepted by the Greedy Algorithm on $S$, we see that $u_k \not \in span(\cup_{i = 1}^{k - 1} \{ u_i \})$. Since $\cup_{i = 1}^{k - 1} \{ u_i \} \supset \cup_{i = 1}^{k - 1} \{ u_i \} \cap T$, we see that $u_k \not \in span(\cup_{i = 1}^{k - 1} \{ u_i \} \cap T)$, which implies that $u_k$ is selected by the Greedy Algorithm on $T$. This shows that $MWB(S) \cap T \subset MWB(T)$ for a general subset $T \subset S$. \\
Now further assume that $T = \{ u_1, \cdots, u_m \}$ for some $m \leq n$. Assume $u_k \in MWB(T)$. Since the Greedy Algorithm selected $u_k$, we know $u_k \not \in span(\cup_{i = 1}^{k - 1} \{ u_i \})$, and therefore, $u_k \in MWB(S)$. 
\end{proof}

\begin{lemma}
\label{lemma:MWBKickOutOne}
For any set $S \subset \U$ and element $u \in \U$ in a matroid $\M = (\U, \I)$, we have
\begin{equation*}
    \abs{ MWB(S) \setminus MWB(S \cup \{ u \}) } \leq 1
\end{equation*}
\end{lemma}
\begin{proof}
Since the size of a basis is equal to the rank of the ground set, we see that
\begin{equation*}
    \abs{MWB(S \cup \{ u \} )} \geq \abs{MWB(S)}
\end{equation*}
Also, by Lemma \ref{lemma:MWBInSubset}, we have $MWB(S \cup \{ u\}) \cap S \subset MWB(S)$. Notice that 
\begin{equation*}
    MWB(S \cup \{ u\}) \cap S = MWB(S \cup \{ u\}) \setminus \{ u \}
\end{equation*}
and therefore
\begin{equation*}
\abs{MWB(S \cup \{ u \} ) \cap S} \geq \abs{MWB(S)} - 1
\end{equation*}
and we have the desired result
\begin{equation*}
    \abs{ MWB(S) \setminus MWB(S \cup \{ u \}) } \leq \abs{ MWB(S) \setminus (MWB(S \cup \{ u \} ) \cap S) } \leq 1
\end{equation*}
\end{proof}

\newpage

\begin{proposition}
\label{prop:MWBisMWBofMWB}
For any set $S \subset \U$ and element $u \in \U$ in a matroid $\M = (\U, \I)$, we have
\begin{equation*}
    MWB(S \cup \{ u\}) = MWB(MWB(S) \cup \{ u \})
\end{equation*}
\end{proposition}
\begin{proof}
Label the elements of 
\begin{equation*}
S \cup \{ u \} = \{ u_1, \cdots, u_k, u, u_{k + 1}, \cdots, u_n \}
\end{equation*}
and 
\begin{equation*}
MWB(S) \cup \{ u \} = \{ u_{i_1}, \cdots, u_{i_{k'}}, u, u_{i_{k'} + 1}, \cdots, u_{i_{n'} + 1} \}
\end{equation*}
in decreasing order of their values. Also, let $T = \{ u_1, \cdots, u_k \}$. Now imagine two greedy algorithms running in parallel, $\A_1$ on $S \cup \{ u \}$ and $\A_2$ on $MWB(S) \cup \{ u \}$. If the two algorithms encounter an element that appears in both sets, they will process the element concurrently; if they encounter an element that appears only in $S \cup \{ u \}$, only $\A_1$ will process that element, while $\A_2$ skips over that index. Let $A_i, B_i$ respectively be the set of elements that $\A_1$ and $\A_2$ have accepted after $\A_1$ processed the element $u_i$. We will prove that $\A_1, \A_2$ make the same decisions on each element index by index.
\paragraph*{1. On $u_1, \cdots, u_k$: }
By Lemma \ref{lemma:MWBInSubset},
\begin{equation*}
    MWB(T) = MWB(MWB(S) \cap T)
\end{equation*}
This means that $A_k = B_k = \{ u_{i_1}, \cdots, u_{i_{k'}} \}$, and the two algorithms make the same decision for all elements before $u$.
\paragraph*{2. On $u$: }
Since $A_k = B_k$, it is obvious that the two algorithms make the same decision for $u$ as well. $u$ will accepted if and only if $A_k \cup \{ u \} \in \I$. If $u$ was rejected, then the two algorithms essentially reduce to finding $MWB(S)$ and $MWB(MWB(S))$ respectively, which have to equal each other. So now assume $u$ was accepted.
\paragraph*{3. On $u_{k + 1}, \cdots, u_n$: }
We will use mathematical induction. Assume that $A_{i - 1} = B_{i - 1}$. If $u_i$ appears in both $S$ and $MWB(S)$, it is obvious that the two algorithms make the same decision for $u_i$. Next consider the case where $u_i$ appears only in $S$. Since $\A_2$ cannot accept this element, $\A_1$ needs to reject this element. Assume to the contrary that $u_i$ is accepted by $\A_1$. Now we will find an index $j < i$ such that $u_j \in MWB(S)$ but $u_j \not \in A_j$ and $u_j \not \in B_j$. By Lemma \ref{lemma:MWBKickOutOne}, there is at most one such index.
\subparagraph*{Case 1 - there is no such $j$: }
This means that
\begin{equation*}
    A_{i - 1} \setminus \{ u \} = MWB(S) \cap \{ u_1, \cdots, u_{i - 1} \}
\end{equation*}
That is, all elements of $MWB(S)$ until the element $u_i$ have all been captured by $A_{i - 1}$.
Then since $A_{i - 1} \cup \{ u_i \}$ is independent, so is $A_{i - 1} \cup \{ u_i \} \setminus \{ u \}$ by downward-closedness. This means that a greedy algorithm should accept $u_i$ into the max-weight basis of $S$, which is a contradiction to the assumption that $u_i \not \in MWB(S)$.
\subparagraph*{Case 2 - there is one such $j$: }
Then 
\begin{equation*}
    A_{i - 1} \cup \{ u_j \} \setminus \{ u \} = MWB(S) \cap \{ u_1, \cdots, u_{i - 1} \}
\end{equation*}
is independent. At the same time $A_{i - 1} \cup \{ u_i \}$ is independent. By augmentation property, at least one of $A_{i - 1} \cup \{ u_j \}$ or $A_{i - 1} \cup \{ u_j, u_i \} \setminus \{ e \}$ should be independent. However, the former case is contradictory to the assumption that $\A_1$ did not accept $u_j$ into $MWB(S \cup \{ u \})$, and the latter case is contradictory to the assumption that a greedy algorithm did not accept $u_i$ into $MWB(S)$. 
\end{proof}
%++++++++++++++++++++++++++++++++++++++++++++++++++++++++++++++++++++++++++++++++++++
%++++++++++++++++++++++++++++++++++++++++++++++++++++++++++++++++++++++++++++++++++++
%++++++++++++++++++++++++++++++++++++++++++++++++++++++++++++++++++++++++++++++++++++
\section{Frameworks for MSP}
In this section, we provide some of the frameworks used to analyze algorithms for MSP in previous works. We first start by introducing a continuous arrival setting for the MSP, where the elements are assumed to arrive at a time independently and randomly drawn. Then we formally define the competitive ratio of an algorithm in a MSP setting. Next we present the result from \cite{BBSW2022} that any ``Greedy'' algorithm (formal definition below) for MSP cannot have a constant competitive ratio. Lastly, we provide the Forbidden Set argument from \cite{STV2021} that identifies a sufficient condition for an algorithm for an ordinal MSP setting to have a constant competitive ratio. We will make use of these frameworks in analyzing some simple algorithms in the next section.
%++++++++++++++++++++++++++++++++++++++++++++++++++++++++++++++++++++++++++++++++++++
\subsection{Continuous Arrival Setting}
For the remaining parts of the paper, consider the continuous arrival setting for MSP: given a matroid $\M = (\U, \I)$ and a weight function $v: \U \rightarrow \R_{\geq 0}$, each element $u \in \U$ arrives at a time $t(u)$ independently and uniformly drawn from $[0, 1]$. The weight of each element is revealed to an algorithm $\A$ when it arrives, and $\A$ needs to make a irrevocable decision to accept the element or not, before it can receive another element. Let $V_t = \{ u \in \U \; \vert \; t(u) < t \}$ be the set of elements that arrive strictly before time $t$. Also, let $A$ be the set of the elements that $\A$ accepts when it terminates and $A_t = \{ u \in A \; \vert \; t(u) < t \}$ be the set of elements that the algorithm accepts strictly before time $t$.

We will assume that $\A$ satisfies the correctness property: it only outputs an independent set; that is $A \in \I$. Additionally, we will only consider algorithms that satisfy the sampling property. Before the algorithm starts, it will choose a fixed \emph{threshold time} $p$ from $[0, 1]$. The algorithm will reject (but store) all elements that arrive in the \emph{sampling phase}, before the threshold time $p$. $p$ will also be referred to as the \emph{sampling probability}. Let $S = V_p$ be the set of samples.
%++++++++++++++++++++++++++++++++++++++++++++++++++++++++++++++++++++++++++++++++++++
\subsection{Competitive Ratio}
Consider an algorithm $\A$ for the MSP, with a matroid $\M = (\U, \I)$ and weight function $v: \U \rightarrow \R_{\geq 0}$. Then the competitive ratio of $\A$ can be defined in two different ways.

\begin{definition}
An algorithm $\A$ has \emph{utility competitive ratio} of $\alpha$ if $\mathbb{E}[v(A)] \geq \alpha \cdot v(MWB(\U))$ where the expectation is taken over the randomness of $\A$
\end{definition}

\begin{definition}
An algorithm $\A$ has \emph{probability competitive ratio} of $\alpha$ if $\Pr[u \in A] \geq \alpha$ for any $u \in MWB(\U)$ where the probability is taken over the randomness of $\A$
\end{definition}

It is obvious that probability competitive ratio is a stronger version of the two: any algorithm that has probability competitive ratio of $\alpha$ also has utility competitive ratio of $\alpha$. In most MSP settings, the competitive ratio is usually evaluated with the utility competitive ratio, but sometimes it is necessary to make use of the probability version. One such setting is when we consider the ordinal version of the MSP.

Because the max-weight basis of a matroid can be found with the greedy algorithm \ref{algo:Greedy}, the actual weights of the elements do not particularly matter; it is only the ordering of the weights that matter when identifying the max-weight basis. Therefore, it is possible to rewrite the MSP such that we are only given the relative weight ordering of the elements, not the actual values of the weights. In such case, it is impossible to evaluate an algorithm with the utility competitive ratio, and we need to rely on the probability competitive ratio.
%++++++++++++++++++++++++++++++++++++++++++++++++++++++++++++++++++++++++++++++++++++
\subsection{Forbidden Sets}
In this part, we present the framework of Forbidden Sets from \cite{STV2021}. For the purpose of this part, we consider the ordinal MSP setting: an algorithm is provided a matroid $\M = (\U, \I)$ and the relative ordering of the weights of the elements, but not the specific values. In such a setting, Soto et al. \cite{STV2021} found a property of an algorithm that is a necessary condition for having a constant probability competitive ratio.

\begin{definition}
An algorithm $\A$ for an ordinal MSP has \emph{Forbidden Sets} of size $k$, if for every triple $(X, Y, u)$ with $Y \subset \U, u \in MWB(Y)$ and $X \subset Y \setminus \{ u \}$, one can define a set $\F(X, Y, u) \subset X$ of at most $k$ \emph{forbidden elements} of $X$ such that the following conditions holds. 
\begin{itemize}
    \item Let $u$ be an element that arrives after the sampling phase. If $u \in MWB(V_{t(u)} \cup \{ u \})$ and for every $u' \in V_{t(u)} \setminus S$, we have $u' \not \in \F(V_{t(u')} \cup \{ u'\}, V_{t(u)} \cup \{ u \}, u)$, then $u \in A$
\end{itemize}
\end{definition}

To explain the Forbidden Set property in simpler words: for any pair of an optimal element $u$ we want to add and the current configuration $Y$ it appeared in, and for each possible previous configurations $X$ for that pair, I want to define a set of ``forbidden elements''. If each of the elements that arrived before $u$ was not one of the forbidden elements for its respective configuration they arrived in, then the algorithm must accept $u$. If it is possible to define a set of forbidden elements for all possible scenarios such that it is consistent with the decisions of the algorithm, the algorithm can achieve a constant probability competitive ratio with a specific choice of sampling probability $p$.

\begin{theorem} [From \cite{STV2021}]
\label{thm:ForbiddenSetsConstant}
If an algorithm $\A$ has Forbidden Sets of size $k$, then it obtains a constant probability competitive ratio of $\alpha(k)$ with the choice of sampling probability $p(k)$ where
\begin{equation*}
    (\alpha(k), p(k)) = 
    \begin{cases}
        \left( \frac{1}{e}, \frac{1}{e} \right) & k = 1 \\
        \left( k^{-\frac{k}{k - 1}}, k^{-\frac{1}{k - 1}} \right) & k \geq 2
    \end{cases}
\end{equation*}
\end{theorem}
%++++++++++++++++++++++++++++++++++++++++++++++++++++++++++++++++++++++++++++++++++++
\subsection{Greedy Algorithms}
In this part, we present the framework of ``Greedy'' Algorithms from \cite{BBSW2022}. This class of algorithms extend the idea of greedy algorithms from Classical or Multiple Choice Secretary Problem into a matroid setting: an element is accepted based on feasibility and optimality. If an element $u$ is a ``good enough'' element, then $u$ will be accepted as long as it is feasible to do so. The formal definition is as follows

\begin{definition}
\label{def:GreedyAlgoMSP}
An algorithm $\A$ is a \emph{Greedy Algorithm for MSP} (or just \emph{Greedy Algorithm} if the meaning can be understood unambiguously) if it satisfies the following conditions
\begin{enumerate}
    \item \label{item:defGreedyAlgoIt} At all times $t \geq p$, maintain an independent set $I_t \in \I$ such that
    \begin{itemize}
        \item $A_t\subset I_t\subset A_t\cup S$
        \item $V_t = span(I_t)$
    \end{itemize}
    \item Accept $u$ if and only if $u \in MWB((\M |_{I_{t(u)}\cup \{u\}}) /  A_{t(u)})$
\end{enumerate}
\end{definition}

$I_t$ in the definition above serves as a reference set, which can be used to decide if a newly arrived element $u$ is good enough. Then, when we contract the matroid with $A_t$, we only consider elements that are feasible. The definition of a Greedy Algorithm is quite extensive. Dynkin's algorithm from section \ref{section:CSP} and Optimistic algorithm from section \ref{section:MCSP}, the algorithms that obtain optimal solutions for the Classical and the Multiple Choice Secretary Problem both fit the definition. However, it still fails to capture similar greedy-like algorithms as we will discuss more in next section.

\begin{remark}
All algorithms we will consider in this paper mimic the basic idea of a Greedy Algorithm: it stores and maintains the max-weight basis of some set $S_t$ at all times $t$ as a reference set. When a new element $u$ arrives at time $t(u)$, then we may update $MWB(S_{t(u)})$ to $MWB(S_{t(u)} \cup \{ u \})$.  Proposition \ref{prop:MWBisMWBofMWB} guarantees that storing just the max-weight basis of $S_t$ is sufficient to encode the information of $S_t$. Then Lemma \ref{lemma:MWBKickOutOne} says that there is at most one element that will be removed from $MWB(S_{t(u)})$, when trying to add the new element. If such an element exists, we informally refer to it as the element that the algorithm \emph{kicks out}. Moreover, by Lemma \ref{lemma:MWBInSubset}, an element that is kicked out is discarded forever and is never included in the max-weight basis of $S_t$ for any future time $t$. In an online algorithm setting, this is a useful property because the amount of data we need to store is upper bounded by a function of $k$, the rank of the matroid, not $n$, the number of elements.
\end{remark}

An interesting result of \cite{BBSW2022} was that any instance of Greedy Algorithm cannot achieve a constant utility competitive ratio for the MSP. The paper introduces the following graph, which works as an adversarial input to any Greedy Algorithm.

\begin{definition}
\label{def:HatGraph}
\emph{Hat Graph} is a graph $G = (V, E)$ with $n + 2$ vertices and $2n + 1$ edges where
\begin{itemize}
    \item $V = \{ v_t, v_b, v_1, v_2, \cdots, v_n \}$ 
    \item $E = \{ e_{\infty} = (v_t, v_b) \} \cup \{ t_i = (v_t, v_i) \; \vert \; i \in [n] \} \cup \{ b_i = (v_b, v_i) \; \vert \; i \in [n] \}$
\end{itemize}
The edge $e_\infty$ is called the \emph{infinity edge} and edges $t_i$ are called \emph{top edges} and the edges $b_i$ are called \emph{bottom edges}. The set $\{ t_i, b_i \}$ of corresponding top and bottom edges is called a \emph{claw}. Additionally, the weight function $v: E \rightarrow \R_{\geq 0}$ is given to the set of edges such that
\begin{equation*}
    v(e_\infty) >> v(t_1) > \cdots > v(t_n) > v(b_1) > \cdots > v(b_n)
\end{equation*}
where the weight of the infinity edge is significantly larger than the sum of the remaining weights.
\end{definition}

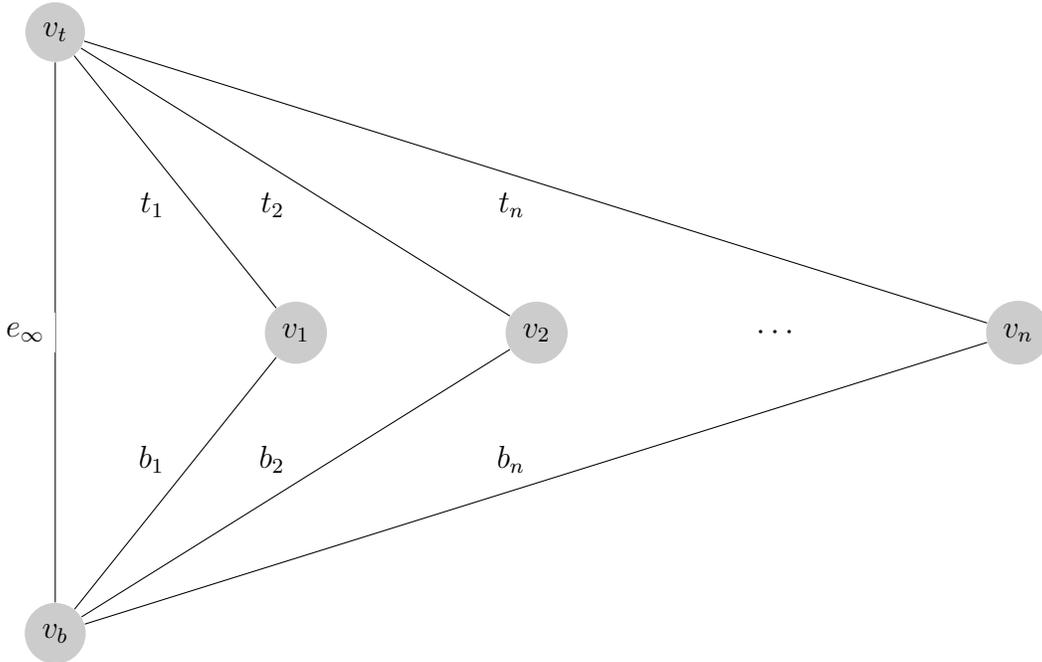
\begin{figure}[h]
\centering
    \begin{tikzpicture} [scale=.8,auto=left]
        \node (vt) at (0, 10) [circle, fill=black!20] {$v_t$};
        \node (vb) at (0, 0)  [circle, fill=black!20] {$v_b$};
        \node (v1) at (4, 5)  [circle, fill=black!20] {$v_1$};
        \node (v2) at (8, 5) [circle, fill=black!20] {$v_2$};
        \node (vn) at (16, 5) [circle, fill=black!20] {$v_n$};
        \node (skip) at (12, 5)  [circle, draw=white!20, fill=white!20] {$\cdots$};
        
        \draw (vb) -- (vt) node [midway, fill=white] {$e_\infty$};      
        \draw (v1) -- (vt) node [midway, fill=white] {$t_1$};
        \draw (v2) -- (vt) node [midway, fill=white] {$t_2$};
        \draw (vn) -- (vt) node [midway, fill=white] {$t_n$};
        \draw (vb) -- (v1) node [midway, fill=white] {$b_1$};
        \draw (vb) -- (v2) node [midway, fill=white] {$b_2$};
        \draw (vb) -- (vn) node [midway, fill=white] {$b_n$};
    \end{tikzpicture}
\caption{Hat Graph as defined in \ref{def:HatGraph}}
\end{figure}

\begin{theorem}[From \cite{BBSW2022}]
A Greedy Algorithm cannot obtain a constant utility competitive ratio on the Hat Graph
\end{theorem}

The full proof will not be included here, but the basic idea of the proof is to show that the probability that $e_\infty$ is rejected by a Greedy Algorithm converges to 1 as $n$ goes to infinity. This same idea will be used in a later section.
%++++++++++++++++++++++++++++++++++++++++++++++++++++++++++++++++++++++++++++++++++++
%++++++++++++++++++++++++++++++++++++++++++++++++++++++++++++++++++++++++++++++++++++
%++++++++++++++++++++++++++++++++++++++++++++++++++++++++++++++++++++++++++++++++++++
\section{Results}
\subsection{Algorithms That Make Decisions Based on Samples}
In this section, we consider two algorithms for MSP, which are defined by slightly modifying the definition for a Greedy Algorithm. Instead of using $I_t$ as a reference set to determine if the newly seen element $u$ is a good enough element, the two will both use just $S$, the set of samples. This resembles the idea of Dynkin's algorithm or the Optimistic Algorithm , where the best elements of $S$ were used as reference elements until the algorithm terminated. 

While both algorithms reduce to Dynkin's algorithm for the Classical Secretary Problem, they act slightly differently on matroids. One of them will check if a newly seen element $u$ is in the max-weight basis when the matroid is contracted by $A_t$, just like in a Greedy Algorithm; the other will only check if $A_t \cup \{ u \}$ is independent. Interestingly, the first algorithm will still be an instance of the Greedy Algorithm with a correct choice of $I_t$. The second algorithm will be proven not to be a Greedy Algorithm. We present the formal definitions of the two algorithms below.

\begin{definition}
Define the algorithm \emph{SAMPLE-CONTRACTED} as the following
\begin{enumerate}
    \item For all $t \geq p$, accept $u$ if and only if $u \in MWB((\M /  A_{t(u)})|_{S \cup \{u\}})$
\end{enumerate}
\end{definition}

\begin{definition}
Define the algorithm \emph{SAMPLE} as the following
\begin{enumerate}
    \item For all $t \geq p$, accept $u$ if and only if 
    \begin{itemize}
        \item $A_{t(u)} \cup \{ u \}$ is independent
        \item $u \in MWB(\M|_{S \cup \{ u \}})$
    \end{itemize}
\end{enumerate}
\end{definition}

On $k$-uniform matroids, the SAMPLE-CONTRACTED algorithm reduces to the Optimistic Algorithm: it uses the $(k - i)$-th heaviest sample as a reference element where $i$ is the number of accepted elements so far. The algorithm SAMPLE on the other hand, is using the weight of the $k$-th heaviest sample as the threshold throughout the entire algorithm.

Now we prove a few lemmas to show that SAMPLE-CONTRACTED is an instance of Greedy Algorithm.

\begin{lemma}
\label{lemma:SampleContracted1}
Let $I_t = MWB((\M / A_t)|_S) \cup A_t$. Then this choice of $I_t$ satisfies the condition \ref{item:defGreedyAlgoIt} of the Definition \ref{def:GreedyAlgoMSP} of a Greedy Algorithm.  
\end{lemma}
\begin{proof}
Let us first prove that $I_t$ is independent. By the definition of a restriction of a matroid, any independent set $I$ in the $(\M / A_t) |_S$ is a subset of a independent set in $\M / A_t$, and by downward-closedness, it is independent in $\M / A_t$. By the definition of a contracted matroid, $I \cup A_t$ is independent in $\M$. Therefore, $I_t \in \I$. Next, it is obvious that $A_t \subset I_t \subset A_t \cup S$. Now we want to prove that $V_t = span(I_t)$. Since $I_t$ is independent, all we need to prove is for any $u \in V_t \setminus I_t$, we have $I_t \cup \{ u \} \not \in \I$. 

First consider the case where $u \in S$. Assume to the contrary that $I_t \cup \{ u \} \in \I$. Let $I = MWB((\M \setminus A_t)|_S)$. Then, $I \cup \{ u \}$ is independent in $\M \setminus A_t$, and since $I \cup \{ u \}$ is entirely contained in $S$, we see that it is also independent in $(\M \setminus A_t)|_S$. This goes against the assumption that $I$ is a basis for the matroid $(\M \setminus A_t)|_S$. 

Now consider the case where $u \not \in S$. Since $u \not \in A_t$, $u$ has to be one of the elements that the algorithm rejected after the time $p$. Then if let $t' := t(u)$, we would have had $u \not \in MWB((\M |_{I_{t'}\cup \{u\}}) /  A_{t'})$. If we let $s_1, \cdots, s_k \in S \cap I_{t'}$ be the elements with larger weight than $u$ (in the decreasing order of their weights), then $u \in span(s_1, \cdots, s_k)$ in the contracted matroid $(\M |_{I_{t'}\cup \{u\}}) /  A_{t'}$. Now consider all elements $a_1, \cdots, a_{k'} \in A_t \setminus A_{t'}$ in the order in which they arrive. When $a_j$ got accepted, we know from an earlier observation that it would have kicked out at most one element of $S$ from $I_{t(a_j)}$. First consider the case where it did not kick out any element of $S$. Then it is clear that $u \in span(s_1, \cdots, s_k, a_j)$ in the same contracted matroid $(\M |_{I_{t'}\cup \{u\}}) /  A_{t'}$. Therefore, $u \in span(s_1, \cdots, s_k)$ in the newly contracted matroid $(\M |_{I_{t'}\cup \{a_j, u\}}) /  (A_{t'} \cup \{ a_j \} )$. Next consider the case where it kicked out the element $s_{i_j} \in S$. Then by Lemma \ref{lemma:SpanTheSame}, we see that $u \in span(s_1, \cdots, s_{i_j - 1}, a_j, s_{i_j + 1}, \cdots,  s_k)$ in the matroid. Therefore, $u \in span(s_1, \cdots, s_{i_j - 1}, s_{i_j + 1}, \cdots, s_k)$ in the newly contracted matroid $(\M |_{I_{t'}\cup \{a_j, u\}}) /  (A_{t'} \cup \{ a_j \} )$. If we apply this same logic for each $a_j$ inductively, we see that $u \in span(\{ s_1, \cdots, s_k \} \cap I_t)$ in the contracted matroid $(\M |_{I_{t}\cup \{u\}}) /  A_{t})$. This shows that $I_t \cup \{ u \}$ is dependent.
\end{proof}

\begin{lemma}
\label{lemma:SampleContracted2}
The Greedy Algorithm with the choice of $I_t = MWB((\M / A_t)|_S) \cup A_t$ is equivalent to the algorithm SAMPLE-CONTRACTED
\end{lemma}
\begin{proof}
This is just an application of Proposition \ref{prop:MWBisMWBofMWB} in a contracted matroid. We just take the max-weight basis of the equation below 
\begin{equation*}
MWB(S \cup \{ u \}) = MWB(MWB(S) \cup \{ u \})
\end{equation*}
to be taken over the contracted matroid $(\M / A_t)$
\begin{equation*}
    MWB((\M / A_t) |_{S \cup \{ u \}} ) = MWB( (\M|_{MWB((\M / A_t)|_S) \cup \{ u \} \cup A_t }) / A_t )
\end{equation*}
If we denote $MWB((\M / A_t)|_S) \cup A_t$ by $I_t$, we easily see that the two algorithms are equivalent.
\end{proof}

Combining the results of Lemma \ref{lemma:SampleContracted1} and Lemma \ref{lemma:SampleContracted2}, we get the following theorem.
\begin{theorem}
SAMPLE-CONTRACTED is a Greedy Algorithm
\end{theorem}

Now we turn to proving that the other algorithm SAMPLE is \emph{not} a Greedy Algorithm.
\begin{theorem}
\label{thm:SampleNotGreedyAlgo}
SAMPLE is not a Greedy Algorithm
\end{theorem}
\begin{proof}
Consider the following undirected graph $G$: $V = \{v_1, v_2, v_3\}$, $E = \{ e_1 = (v_1, v_2), e_2 = (v_2, v_3), e_3 = (v_3, v_1) \}$ with the weights $v(e_i) = i$ for each $i$. Assume that the edge $e_3$ arrives during the sampling phase, and $e_2, e_1$ arrive after the sampling phase, in that order. Any instance of Greedy Algorithm, and the algorithm SAMPLE will accept $e_2$. But SAMPLE accepts the last element $e_1$, whereas a Greedy Algorithm will reject it. To see why a Greedy Algorithm should reject it, consider what $I_t$ will have to be when $t 
:= t(e_1)$.  Since $I_t$ should span $V_t$, we are required to store $I_t = V_t$. This means that
\begin{equation*}
    MWB((\M|_{I_t \cup \{ e_1 \}})/ A_t) = \{ e_3 \}
\end{equation*}
and the Greedy Algorithm will reject $e_1$.
\end{proof}

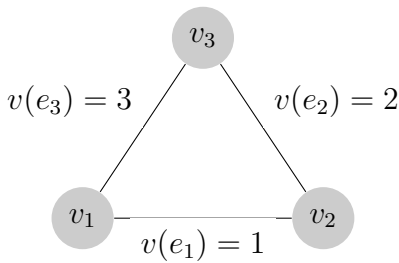
\begin{figure}[h]
\centering
    \begin{tikzpicture} [scale=.8,auto=left]
        \node (v1) at (5, 0)  [circle, fill=black!20] {$v_1$};
        \node (v2) at (9, 0)  [circle, fill=black!20] {$v_2$};
        \node (v3) at (7, 3)  [circle, fill=black!20] {$v_3$};
            
        \draw (v2) -- (v1) node [midway, fill=white] {$v(e_1) = 1$};
        \draw (v3) -- (v2) node [midway, fill=white] {$v(e_2) = 2$};
        \draw (v1) -- (v3) node [midway, fill=white] {$v(e_3) = 3$};
    \end{tikzpicture}
\caption{The undirected graph in the proof of Theorem \ref{thm:SampleNotGreedyAlgo}}
\end{figure}

%++++++++++++++++++++++++++++++++++++++++++++++++++++++++++++++++++++++++++++++++++++
\subsection{Virtual Algorithm for MSP}
In the previous section, we saw two algorithms that were inspired by Dynkin's algorithm, where they used the sample set $S$ as the reference set instead of $I_t$ in the definition of the Greedy Algorithms. In this section, we turn our attention to the Virtual Algorithm from Section \ref{section:MCSP}. We first formally define its generalized version for the MSP. We then show that it is \emph{not} a Greedy Algorithm. In the next two sections, we will analyze more in depth about how the generalized Virtual Algorithm performs on the Hat Graph and a modified version of the graph.

\begin{definition}
\label{algo:VirtualMSP}
The following algorithm is defined as the \emph{Virtual Algorithm for MSP} (or just \emph{Virtual Algorithm} if the meaning can be understood unambiguously) 
\begin{enumerate}
    \item Accept $u$ if and only if
    \begin{itemize}
        \item $A_t \cup \{ u \}$ is independent in $\M$
        \item $u \in MWB(V_t \cup \{ u \})$
        \item $MWB(V_t) \setminus MWB(V_t \cup \{ u \})$ is empty or contains exactly one element of $S$
    \end{itemize} 
\end{enumerate}
\end{definition}

It is easy to verify that the Virtual Algorithm for the MSP is equivalent to the Virtual Algorithm on $k$-uniform matroids. Now that we have formally defined the generalized version of the Virtual Algorithm, we prove that it is \emph{not} a Greedy Algorithm. We show this by first proving that all instances of Greedy Algorithm satisfy the property that they need to choose the top element seen so far as long as it is possible to do so, then proving that the Virtual Algorithm does not satisfy the property.

\begin{proposition}
\label{prop:chooseTopElement}
Let $\A$ be a Greedy Algorithm. For any element $u \in \U$, let $t = t(u)$ be the time it arrives. If $u$ has the largest weight of the elements in $V_t$, $\A$ will accept it if and only if $A_t \cup \{ u \}$ is independent in $\M$
\end{proposition}
\begin{proof}
Assume $A_t \cup \{ u \}$ is independent in $\M$. Then $\{ u \}$ is independent in $\M / A_t$. Since $u$ has the largest weight out of all elements in $V_t$, in particular, it has the largest weight out of all elements in $I_t \cup \{ u \}$. Since a greedy algorithm is guaranteed to find the max-weight basis of a matroid, we easily observe that $u$ is in the max-weight basis of $I_t \cup \{ u \}$ in the contracted matroid $\M / A_t$. Then $\A$ will accept $u$.

On the other hand, assume $\A$ accepted $u$. Then we know that $u$ is in the max-weight basis of $I_t \cup \{ u \}$ in $\M / A_t$. Then by downward-closedness of a matroid, $\{ u \}$ is independent in the same contracted matroid. This shows that $A_t \cup \{ u \}$ is independent in $\M$. 
\end{proof}

\begin{lemma}
\label{lemma:VirtualChooseTopElement}
Virtual Algorithm does not satisfy the property in Proposition \ref{prop:chooseTopElement}. 
\end{lemma}
\begin{proof}
Consider a 2-uniform matroid $\M = (\U, \I)$ on the universe $\U = \{ 1, 2, 3, 4, 5, 6 \} $, where each element of $\U$ is given the weight equal to itself. Assume that the stream of input is given as $(1, 3, 2, 4, 5, 6)$ and that the Virtual Algorithm put aside $S = \{ 1, 3 \}$ as the sample initially. Upon receiving the input $2$, it will accept $2$ and maintain $\{ 2, 3 \}$ as the set of top two elements. Then when the element $4$ arrives, the algorithm rejects it because the element it kicks out is 2 and is not from the sample $S$. But notice that $4$ is the top element seen so far and that $\{ 2, 4 \}$ is independent in $\M$.
\end{proof}

By combining the results of Proposition \ref{prop:chooseTopElement} and Lemma \ref{lemma:VirtualChooseTopElement}, we get the following result.
\begin{theorem}
Virtual Algorithm is not a Greedy Algorithm
\end{theorem}

%++++++++++++++++++++++++++++++++++++++++++++++++++++++++++++++++++++++++++++++++++++
\subsection{Virtual Algorithm on Hat Graph}
We previously proved that the Virtual Algorithm is not an instance of Greedy Algorithm. In this section, we show that the Virtual Algorithm obtains a constant probability competitive ratio on the Hat Graph, unlike any instance of Greedy Algorithm. This fact separates the Virtual Algorithm from the class of Greedy Algorithms in a non-trivial manner.

The proof of the theorem will be presented in two main parts. The first part will show that the edge $e_\infty$ is accepted if $t_1, b_1$, the two edges from the leftmost claw were sampled. This part, by itself, shows that $e_\infty$ is accepted with constant probability. If the weight of the $e_\infty$ is large enough that the weights of all other edges are negligible, this part is enough to show that the algorithm obtains a constant utility competitive ratio. The second part applies the Forbidden Set argument to formally prove that the algorithm obtains a probability competitive ratio of 1/4 on the graph.

\subsubsection{Claw Blocker Argument}
In this part, we will prove that if $t_1, b_1$ are sampled, then the Virtual Algorithm has to accept $e_\infty$. First, we will show that if both edges from the leftmost claw are sampled, the Virtual Algorithm will accept $e_\infty$ as long as it is possible to do so (i.e., $A_{t(e_\infty)} \cup \{e_\infty \}$ is independent). Second, we will show that under the same assumption, the Virtual Algorithm will not accept both edges from any other claw (i.e., it ``blocks'' any other claw from being accepted), which is the only way that $A_{t(e_{\infty})} \cup \{ e_\infty \}$ could be dependent. Throughout the proof, the assumption that $e_\infty$ is not sampled may be implied and not explicitly stated, since it is trivial that if it is sampled, it will not be accepted. 

\begin{lemma}
\label{lemma:VirtualHatGraphConstant1}
If $t_1, b_1 \in S$, the Virtual Algorithm will accept $e_\infty$ as long as it is possible to do so.
\end{lemma}
\begin{proof}
Denote $t(e_\infty)$ as $t$. Consider the max-weight basis of $V_t$. Following the steps of a greedy algorithm, it is clear that $MWB(V_t)$ will include all top edges of $V_t$, and all bottom edges of the claws where the top edges are not in $V_t$, and the single leftmost bottom edge from any claw where both top and bottom edges are in $V_t$. That is
\begin{equation*}
MWB(V_t) = \{ t_i \; \vert \; t_i \in V_t \} \cup \{ b_i : i = \min \{ i \; \vert \; t_i, b_i \in V_t  \} \} \cup \{ b_i \in V_t : t_i \not \in V_t \}
\end{equation*}
Since we assumed that $t_1, b_1 \in S \subset V_t$, we see that $\min \{ i \; \vert \; t_i, b_i \in V_t  \} = 1$. Now consider the max-weight basis of $V_t \cup \{ e_\infty \}$. By Lemma \ref{lemma:MWBInSubset}, this is equal to $MWB(MWB(V_t) \cup \{ e_\infty \} )$. Since $\{ e_\infty, t_1, b_1 \}$ form a cycle, we know that at least one of these three elements need to be kicked out of the max-weight basis. It is clear that $e_\infty, t_1$ are both in the max-weight basis because they are the first two elements to be examined by the greedy algorithm and do not form a cycle with any other elements until $b_1$ is processed. Therefore, we see that $b_1 \in S$ is the element being kicked out by the greedy algorithm. Since the element that was kicked out was a sample, the Virtual Algorithm will accept $e_\infty$
\end{proof}

\begin{lemma}
\label{lemma:VirtualHatGraphConstant2}
If $t_1, b_1 \in S$, it is impossible to have $t_i, b_i \in A_{t(e_\infty)}$ for any $i \in [n]$
\end{lemma}
\begin{proof}
The statement is obviously true for $i = 1$, so assume $i \not = 1$. First assume that $b_i$ arrives after $t_i$ was accepted. Since $\min \{ i \; \vert \; t_i, b_i \in V_{t(b_i)}  \} = 1$, we notice that $b_i \not \in MWB(V_{t(b_i)} \cup \{ b_i \})$. Therefore, the virtual algorithm will reject $b_i$. Now assume that $t_i$ arrives after $b_i$ was accepted. Then notice that $t_1, b_1, b_i \in MWB(V_{t(t_i)})$. However, after we add $t_i$ to $V_{t(t_i)}$, notice that $b_i$ has to be kicked out of the max-weight basis. Since $b_i \not \in S$, the Virtual Algorithm will reject $t_i$. 
\end{proof}
We are now ready to prove the first part of the main theorem.
\begin{proposition}
\label{prop:VirtualHatGraphConstant}
Virtual Algorithm accepts $e_\infty$ with constant probability
\end{proposition}
\begin{proof}
If $t_1, b_1 \in S$, then Lemma \ref{lemma:VirtualHatGraphConstant2} shows that the the leftmost claw blocks any other claw from being accepted. Therefore, when $e_\infty$ arrives, it will be possible to add it to the accepted elements. Lemma $\ref{lemma:VirtualHatGraphConstant1}$ now guarantees that the Virtual Algorithm will indeed accept it. Then, the probability that $e_\infty$ is accepted is lower bounded by the probability that $t_1, b_1 \in S$
\begin{align*}
    \Pr[e_{\infty} \in A_1] 
    &= \Pr[e_\infty \not \in S] \cdot \Pr[e_\infty \in A_1 \; \vert \; e_\infty \not \in S] \\
    &\geq \Pr[e_\infty \not \in S] \cdot \Pr[t_1, b_1 \in S \; \vert \; e_\infty \not \in S] \\
    &= p^2(1 - p)
\end{align*}
where $p$ is the sample probability that was fixed before the algorithm started.
\end{proof}

\subsubsection{Forbidden Sets Argument}
\label{section:VirtualHatForbiddenSets}
The Forbidden Sets property requires that we define a set $\mathcal{F}(X, Y, u)$ for every tuple $(X, Y, u)$ such that $u \in MWB(Y)$ and $X \subset Y \setminus \{ u \}$. But the Forbidden Sets we will use in the following proof satisfies the additional property that $\mathcal{F}(X, Y, u) = \mathcal{F}(X', Y, u)$ for any $X, X' \subset Y \setminus \{ u \}$. Therefore, for the purpose of this proof, we will abuse the notation and denote $\mathcal{F}(*, Y, u)$ as $\mathcal{F}(Y, u)$.

This additional property simplifies the conceptual meaning of a forbidden set. Now we are defining a set $\F(Y, u)$ such that $u$ will be accepted if it arrives at time $t$ as long as $V_t = Y$ and any element $u' \in \F(Y, u)$ was seen during the sampling phase. 
\begin{proposition}
\label{prop:VirtualAlgoForbiddenSets}
Virtual Algorithm has Forbidden Sets of size 2 for the Hat Graph
\end{proposition}
\begin{proof}
There are five cases we will consider based on the edge $u$ and the set $Y$:
\begin{equation*}
    \F(Y, u) = 
    \begin{cases}
        \{ t_1, b_1 \} & u = e_\infty \\
        \{ b_j \; \vert \; j = \min \{ j > i \; \vert \; t_j, b_j \in Y  \} \} \} & u = t_i, Y \not \ni e_\infty, i = \min \{ i \; \vert \; t_i, b_i \in Y  \} \} \\
        \{ b_i \} & u = t_i, otherwise \\
        \{ b_j \; \vert \; j = \min \{ j > i \; \vert \; t_j, b_j \in Y  \} \} \} & u = b_i, Y \not \ni e_\infty, i = \min \{ i \; \vert \; t_i, b_i \in Y  \} \} \\
        \emptyset & u = b_i, Y \not \ni t_i
    \end{cases}
\end{equation*}
First consider the case where $u = e_\infty$. Proposition \ref{prop:VirtualHatGraphConstant} shows that if $t_1, b_1 \in S$, then $e_\infty$ will be accepted by the Virtual Algorithm. This shows that $t_1, b_1$ are forbidden elements for $e_\infty$. 

Next consider the two cases where $u = t_i$. First, if adding $u$ makes the $i$-th claw the leftmost path between $v_t$ and $v_b$, which is precisely when $e_\infty \not \in V_{t(u)} \cup \{ u \}$ and $i = \min \{ i \; \vert \; t_i, b_i \in V_{t(u)} \cup \{ u \}  \} \}$, then we need to kick out the bottom edge of the previously leftmost claw in $V_{t(u)}$, if it exists. In any other case where $u = t_i$, $u$ doesn't need to kick out any element, or if it does, it will kick out $b_i$.

Next consider the two cases where $u = b_i$. The first is when adding $u$ makes the $i$-th claw the leftmost path between $v_t$ and $v_b$, which is precisely when $e_\infty \not \in V_{t(u)} \cup \{ u \}$ and $i = \min \{ i \; \vert \; t_i, b_i \in V_{t(u)} \cup \{ u \}  \} \}$. With the same logic from the previous case with $u = t_i$, we need to kick out the bottom edge of the previously leftmost claw in $V_{t(u)}$, if it exists. Otherwise, adding $u$ does not create a cycle, and $u$ does not need to kick out an element.
\end{proof}

By Theorem \ref{thm:ForbiddenSetsConstant} and Proposition \ref{prop:VirtualAlgoForbiddenSets}, we have the following result.
\begin{theorem}
\label{thm:VirtualHatGraphConstant}
Virtual Algorithm obtains a constant probability competitive ratio of 1/4 on the Hat Graph.
\end{theorem}
%++++++++++++++++++++++++++++++++++++++++++++++++++++++++++++++++++++++++++++++++++++
\subsection{Virtual Algorithm on Modified Hat Graph}
In the previous section, we showed that the Virtual Algorithm obtains a constant probability competitive ratio on the Hat Graph. This result may lead us to believe that the Virtual Algorithm is able to obtain a constant competitive ratio for all graphic matroids, if not all matroids. However, in this section, we show that when a slight modification is introduced to the Hat Graph, the Virtual Algorithm fails to even obtain a constant utility competitive ratio.

\begin{definition}
\label{def:ModifiedHatGraph}
The \emph{Modified Hat Graph} is a graph $G = (V, E)$ with $2n + 2$ vertices and $4n + 1$ edges where
\begin{itemize}
    \item $V = \{ v_t, v_b, v_{1, 1}, v_{1, 2}, v_{2, 1}, v_{2, 2}, \cdots, v_{n, 1}, v_{n, 2} \}$ 
    \item $E = \{ e_{\infty} = (v_t, v_b) \} \cup \{ 1_i = (v_t, v_{i, 2}) \; \vert \; i \in [n] \} \cup \{ 2_i = (v_t, v_{i, 1}) \; \vert \; i \in [n] \} \\ \cup \{ 3_i = (v_{i, 1}, v_{i, 2}) \; \vert \; i \in [n] \} \cup \{ 4_i = (v_b, v_{i, 2}) \; \vert \; i \in [n] \}$
\end{itemize}
The edge $e_\infty$ is still called the \emph{infinity edge} and the set $\{ 1_i, 2_i, 3_i, 4_i \}$ of corresponding edges is also called a \emph{claw}. Additionally, the weight function $v: E \rightarrow \R_{\geq 0}$ is given to the set of edges such that
\begin{equation*}
    v(e_\infty) > v(4_1) > \cdots > v(4_n) > v(3_1) > \cdots > v(3_n) > v(2_1) > \cdots > v(2_n) > v(1_1) > \cdots > v(1_n)
\end{equation*}
where the weight of the infinity edge is significantly larger than the sum of the remaining weights.
\end{definition}

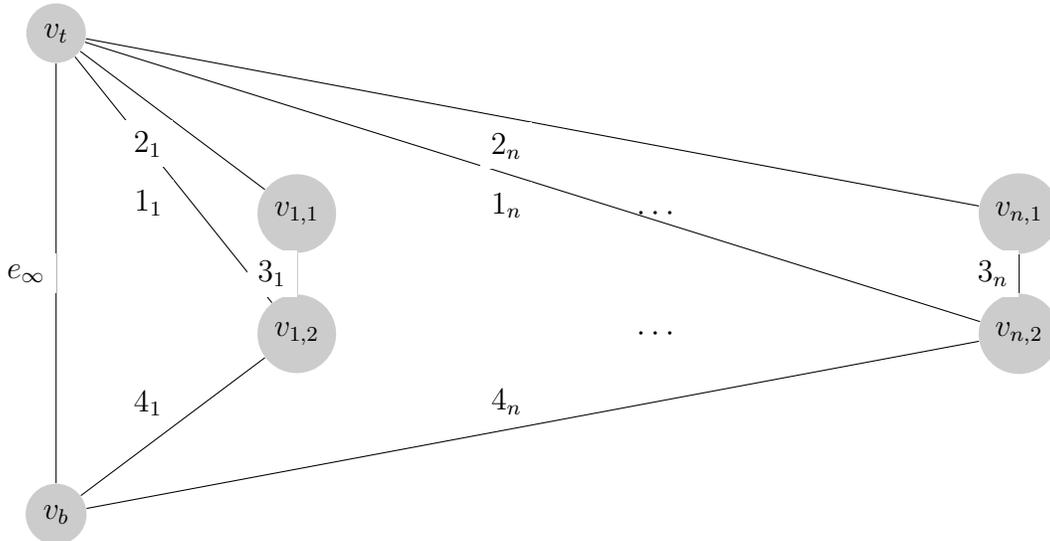
\begin{figure}[h]
\centering
    \begin{tikzpicture} [scale=.8,auto=left]
        \node (vt) at (0, 8) [circle, fill=black!20] {$v_t$};
        \node (vb) at (0, 0)  [circle, fill=black!20] {$v_b$};
        \node (v11) at (4, 5)  [circle, fill=black!20] {$v_{1, 1}$};
        \node (v12) at (4, 3)  [circle, fill=black!20] {$v_{1, 2}$};
        \node (vn1) at (16, 5)  [circle, fill=black!20] {$v_{n, 1}$};
        \node (vn2) at (16, 3)  [circle, fill=black!20] {$v_{n, 2}$};
        \node (skip) at (10, 5)  [circle, draw=white!20, fill=white!20] {$\cdots$};
        \node (skip) at (10, 3)  [circle, draw=white!20, fill=white!20] {$\cdots$};
        
        \draw (vb) -- (vt) node [midway, fill=white] {$e_\infty$};
        \draw (v12) -- (vt)  node [midway, fill=white] {$1_1$};
        \draw (v11) -- (vt)  node [midway, fill=white] {$2_1$};
        \draw (v12) -- (v11) node [midway, fill=white] {$3_1$};
        \draw (vb) -- (v12)  node [midway, fill=white] {$4_1$};
        \draw (vn2) -- (vt)  node [midway, fill=white] {$1_n$};
        \draw (vn1) -- (vt)  node [midway, fill=white] {$2_n$};
        \draw (vn2) -- (vn1) node [midway, fill=white] {$3_n$};
        \draw (vb) -- (vn2)  node [midway, fill=white] {$4_n$};
    \end{tikzpicture}
\caption{Modified Hat Graph as defined in \ref{def:ModifiedHatGraph}}
\end{figure}

\begin{lemma}
\label{lemma:VirtualModifiedHatGraphNotConstant1}
For some $i \in [n]$, if the following conditions hold, then the Virtual Algorithm accepts $1_i$ and $4_i$
\begin{enumerate}
    \item $2_i \in S$ and $1_i, 3_i, 4_i, e_\infty$ arrive after the sampling phase, in that specific order
    \item There exists $j < i$ such that $2_j, 3_j, 4_j \in S$
\end{enumerate}
\end{lemma}
\begin{proof}
Let $t_1, t_3, t_4$ respectively be the times that $1_i, 3_i, 4_i$ arrive. First, it is clear that $1_i \in MWB(V_{t_1} \cup \{ 1_i \}) $ and that adding it will not kick out any element from $MWB(V_{t_1})$. Therefore, the Virtual Algorithm will accept $1_i$. Next, notice that $1_i, 2_i$ are never kicked out of the max-weight basis until $3_i$ arrives. And when it does, $1_i, 2_i, 3_i$ form a cycle, so at time $t_3$, $1_i$ will have to be kicked out of $MWB(V_{t_3})$. Since this element is not a sample, $3_i$ is rejected but is added to the max-weight basis, replacing $1_i$, and $3_i$ will not be kicked out until $4_i$ arrives. Finally, when $4_i$ arrives, $MWB(V_{t_4})$ contains exactly one path from $v_t$ to $v_b$ on the left hand side of the claw $i$. But this path cannot be of the type $\{ 1_k, 4_k \}$ because we already know that $\{ 2_j, 3_j, 4_j \} \in S \subset V_{t_4}$. Therefore, the path has to be of the type $\{ 2_k, 3_k, 4_k \}$ for some $k < i$. Then $2_k, 3_k, 4_k, 2_i, 3_i, 4_i$ form a cycle, and $4_i$ kicks out $2_i$ from the max-weight basis. Since $2_i$ was from the sample, the Virtual Algorithm will accept $4_i$.
\end{proof}

The conditions given in the lemma above are not tight at all. However, they are sufficient for our analysis.
\begin{lemma}
\label{lemma:VirtualModifiedHatGraphNotConstant2}
There exists $j \in [n]$ such that $j \leq \lfloor \frac{n}{2} \rfloor$ and $2_j, 3_j, 4_j \in S$, with probability $p_n = 1 - (1 - p^3)^{\lfloor n/2 \rfloor}$
\end{lemma}
\begin{proof}
For each $j$, the event $2_j, 3_j, 4_j \in S$ happens with probability $p^3$. The probability that one of these events happening is $1 - (1 - p^3)^{\lfloor n/2 \rfloor}$
\end{proof}

\begin{lemma}
\label{lemma:VirtualModifiedHatGraphNotConstant3}
Conditioned on time $t > p$ that the edge $e_\infty$ arrives, there exists $i \in [n]$ such that $i > \lfloor \frac{n}{2} \rfloor$ and $2_i \in S$ and $1_i, 3_i, 4_i, e_\infty$ arrive after the sampling phase, in that specific order, with probability $q_{n, t} = 1 - \left(1 - \left( \frac{p(t - p)^3}{6} \right) \right)^{\lfloor n/2 \rfloor}$
\end{lemma}
\begin{proof}
The probability that $2_i \in S$ is $p$. The probability that each of $1_i, 3_i, 4_i$ arrives after the sampling phase but before $e_\infty$ is $t - p$. The probability that $1_i, 3_i, 4_i$ are permuted in that specific order is $\frac{1}{6}$. Therefore, the probability that the condition is satisfied for one specific $i$ is $\frac{p(t - p)^3}{6}$. The probability that the condition is satisfied for at least one $i$ is $1 - \left(1 - \left( \frac{p(t - p)^3}{6} \right) \right)^{\lfloor n/2 \rfloor}$
\end{proof}

The following lemma is more of a quick observation.
\begin{lemma}
$\underset{n \rightarrow \infty}{\lim} p_n = \underset{n \rightarrow \infty}{\lim} q_{n, t} = 1$ for any $t \in (p, 1]$
\end{lemma}

We are now ready to prove the main result of this part.

\begin{theorem}
The Virtual Algorithm does not obtain a constant competitive ratio on the Modified Hat Graph
\end{theorem}
\begin{proof}
Conditioned on time $t > p$ that the edge $e_\infty$ arrives, if the conditions of Lemma \ref{lemma:VirtualModifiedHatGraphNotConstant2} and the Lemma \ref{lemma:VirtualModifiedHatGraphNotConstant3} are both satisfied, then the two conditions of Lemma \ref{lemma:VirtualModifiedHatGraphNotConstant1} are satisfied. The probability of such event is at least $p_n q_{n, t}$. And when such an even happens, $e_\infty$ cannot be accepted. Also, if $t < p$, then $e_\infty \in S$ and cannot be accepted either.
\begin{align*}
\underset{n \rightarrow \infty}{\lim} \Pr[e_\infty \text{ rejected}] 
&= \underset{n \rightarrow \infty}{\lim} \left( \int_0^p dt + \int_p^1 p_n q_{n, t} dt \right) \\ 
&= \underset{n \rightarrow \infty}{\lim} \left( \int_0^p dt + p_n \int_p^1 q_{n, t} dt \right) \\
&= p + \underset{n \rightarrow \infty}{\lim} p_n \cdot \underset{n \rightarrow \infty}{\lim} \int_p^1 q_{n, t} dt \\
&= p + \underset{n \rightarrow \infty}{\lim} p_n \cdot \int_p^1 \underset{n \rightarrow \infty}{\lim}  q_{n, t} dt \\
&= p + 1 \cdot (1 - p) = 1
\end{align*}
Since $0 \leq q_{n, t} \leq q_{n+1, t}$ for any $n, t$, the Monotone Convergence Theorem guarantees that the limit and the integral interchange between the third and the fourth lines of the equation above.
\end{proof}

%++++++++++++++++++++++++++++++++++++++++++++++++++++++++++++++++++++++++++++++++++++
\subsection{More About Forbidden Sets}
In Section \ref{section:VirtualHatForbiddenSets}, we saw that the Virtual Algorithm has Forbidden Sets of size 2 for the Hat Graph. In general, \cite{STV2021} proved that there is an algorithm that has Forbidden Sets of size 2, thus obtaining a probability competitive ratio of $\frac{1}{4}$, for \emph{any} graphic matroid. A natural question is whether there is an algorithm that has Forbidden Sets of size 1 for any graphic matroid because that would mean that it obtains a probability competitive ratio of $\frac{1}{e}$, the most optimal ratio it can achieve.

In this part, we formally define the \emph{Strong} Forbidden Sets property to be the property that $\F(X, Y, u) = \F(X', Y, u)$ for any $X, X'$. This was the additional property that the Forbidden Sets for the Virtual Algorithm had, and this property allows us to have a better intuitive understanding of the Forbidden Sets argument. We prove that there cannot be an algorithm with Strong Forbidden Sets of size 1 on a general graphic matroid.

\begin{definition}
	An algorithm $\A$ for an ordinal MSP has \emph{Strong Forbidden Sets} of size $k$, if for every tuple $(Y, u)$ with $Y \subset \U, u \in MWB(Y)$, one can define a set $\F(Y, u) \subset Y \setminus \{ u \}$ of at most $k$ \emph{forbidden elements} of $Y$ such that the following conditions holds. 
	\begin{itemize}
		\item Let $u$ be an element that arrives after the sampling phase. If $u \in MWB(V_{t(u)} \cup \{ u \})$ and for every $u' \in V_{t(u)} \setminus S$, we have $u' \not \in \F( V_{t(u)} \cup \{ u \}, u)$, then $u \in A$
	\end{itemize}
\end{definition}

We present a very short lemma.
\begin{lemma}
	\label{lemma:ForbiddenFirstAfterSample}
	If an algorithm $\A$ for an ordinal MSP has the Forbidden Sets property, then it always accepts the first element after the sampling phase
\end{lemma}
\begin{proof}
	Let $u$ be the first element after the sampling phase. Then $V_{t(u)} \setminus S$ is empty, and the condition for Forbidden Sets becomes vacuously true, so $\A$ needs to accept $u$
\end{proof}

\begin{theorem}
	\label{thm:StrongForbiddenSetsSize1}
	There is no algorithm $\A$ for an ordinal MSP with Strong Forbidden Sets of size 1 on all graphic matroids
\end{theorem}
\begin{proof}
	Consider the following undirected graph $G$: $V = \{v_0, v_1, v_2\}$, $E = \{ e_{1, 1}, e_{1, 2} = (v_1, v_2), e_{2, 1}, e_{2, 2} = (v_2, v_3), e_{3, 1}, e_{3, 2} = (v_3, v_1) \}$ with the weights $v(e_{i, j}) = i + 3(j - 1)$ for each $i, j$. For each $i$, we will show that $\F(Y, e_{i, 2})$ has to be $e_{i, 1}$ for any $Y \ni e_{i, 1}$. Assume otherwise, and let $\F(Y, e_{i, 2}) = e'$ for some $Y$ and $e' \not = e_{i, 1}$. Consider the case where $e'$ is sampled, and $e_{i, 1}$ is the first to arrive after the sampling phase, followed by the remaining elements of $Y$, with $e_{i, 2}$ being the last element to arrive from $Y$. Then by Lemma \ref{lemma:ForbiddenFirstAfterSample}, $e_{i, 1}$ is accepted. Also, since $e' \in S$, $e_{i, 2}$ is accepted. This contradicts the correctness property of $\A$. This shows that $\F(Y, e_{i, 2})$ has to be $e_{i, 1}$ for any $Y \ni e_{i, 1}$.
	
	Now consider the case where $e_{1, 1}, e_{2, 1}, e_{3, 1}$ are sampled and the other edges arrive in the following order: $e_{1, 2}, e_{2, 2}, e_{3, 2}$. Then by the Forbidden Sets rule, $\A$ has to accept the three edges, which form a cycle. This contradicts the correctness property of $\A$. 
\end{proof}

Theorem \ref{thm:StrongForbiddenSetsSize1} above shows that there is no algorithm with Strong Forbidden Sets of size 1 on a general graphic matroid. However, we conjecture that there is an algorithm with Strong Forbidden Sets of size 2 on all graphic matroids. On the other hand, we conjecture that there is no algorithm with Forbidden Sets of size 1 on a general graphic matroid. 

\begin{figure}[h]
	\centering
	\begin{tikzpicture} [scale=.8,auto=left]
	\node (v1) at (0, 0)  [circle, fill=black!20] {$v_1$};
	\node (v2) at (12, 0)  [circle, fill=black!20] {$v_2$};
	\node (v3) at (6, 9)  [circle, fill=black!20] {$v_3$};
	
	\draw (v2) to[bend left] node[align=center, fill=white]{$v(e_{1, 1}) = 1$} (v1);
	\draw (v3) to[bend left] node[align=center, fill=white]{$v(e_{2, 1}) = 2$} (v2);
	\draw (v1) to[bend left] node[align=center, fill=white]{$v(e_{3, 1}) = 3$} (v3);
	\draw (v2) to[bend right] node[align=center, fill=white]{$v(e_{1, 2}) = 4$} (v1);
	\draw (v3) to[bend right] node[align=center, fill=white]{$v(e_{2, 2}) = 5$} (v2);
	\draw (v1) to[bend right] node[align=center, fill=white]{$v(e_{3, 2}) = 6$} (v3);
	
	\end{tikzpicture}
	\caption{The undirected graph in the proof of Theorem \ref{thm:StrongForbiddenSetsSize1}}
\end{figure}
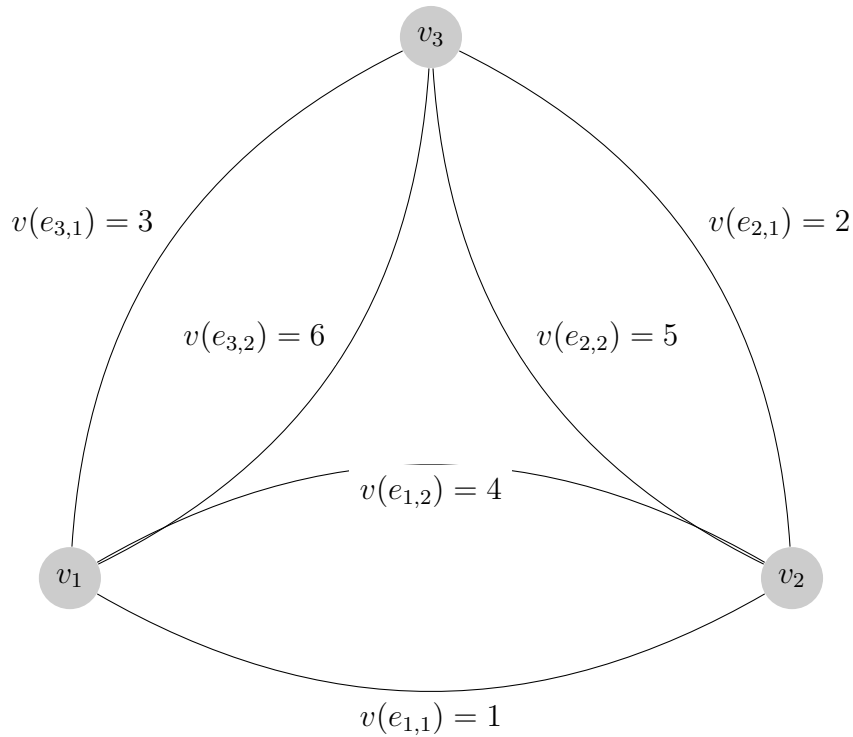
%++++++++++++++++++++++++++++++++++++++++++++++++++++++++++++++++++++++++++++++++++++
%++++++++++++++++++++++++++++++++++++++++++++++++++++++++++++++++++++++++++++++++++++
%++++++++++++++++++++++++++++++++++++++++++++++++++++++++++++++++++++++++++++++++++++
\section{Conclusion}
This paper studies the extensions of some simple algorithms from Classical and the Multiple Choice Secretary Problem to the MSP. We notice that subtle differences in the definitions of the algorithms create a non-trivial difference in the properties and the performance of the algorithms. In future studies, we hope to establish a framework to analyze a \emph{class} of algorithms and their collective properties.

\newpage
\bibliographystyle{abbrv}
\bibliography{bib}

\end{document}